\documentclass[1p,final]{elsarticle}
\usepackage{amsfonts,color,morefloats,pslatex}
\usepackage{amssymb,amsthm, amsmath,latexsym}

\allowdisplaybreaks[4]

\newtheorem{theorem}{Theorem}
\newtheorem{lemma}[theorem]{Lemma}

\newtheorem{example}{Example}

\newcommand\myatop[2]{\genfrac{}{}{0pt}{}{#1}{#2}}

\newcommand{\tr}{{\mathrm{Tr}}}

\newcommand{\gf}{{\mathrm{GF}}}
\newcommand{\PG}{{\mathrm{PG}}}

\newcommand{\wt}{{\mathtt{wt}}}

\newcommand{\bC}{{\mathbb{C}}}

\newcommand{\C}{{\mathcal{C}}}
\newcommand{\cA}{{\mathcal{A}}}

\newcommand{\cH}{{\mathcal{H}}}

\newcommand{\cO}{{\mathcal{O}}}

\newcommand{\bc}{{\mathbf{c}}}

\newcommand{\bh}{{\mathbf{h}}}

\begin{document}

\begin{frontmatter}



\title{The Subfield Codes of  Hyperoval and Conic codes\tnotetext[fn1]{C. Ding's research was supported by
The Hong Kong Research Grants Council, Proj. No. 16300415.}}


\author[zl]{Ziling Heng}
\address[zl]{Department of Computer Science and Engineering, The Hong Kong University of Science and Technology, Clear Water Bay, Kowloon, Hong Kong, China}
\ead{zilingheng@163.com}
\author[cd]{Cunsheng Ding}
\address[cd]{Department of Computer Science and Engineering, The Hong Kong University of Science and Technology, Clear Water Bay, Kowloon, Hong Kong, China}
\ead{cding@ust.hk}



\begin{abstract}
Hyperovals in $\PG(2,\gf(q))$ with even $q$ are maximal arcs and an interesting research topic in finite geometries and combinatorics. Hyperovals in $\PG(2,\gf(q))$ are equivalent to $[q+2,3,q]$ 
MDS codes over $\gf(q)$, called hyperoval codes, in the sense that one can be constructed from the other. Ovals in $\PG(2,\gf(q))$ for odd $q$ are equivalent to $[q+1,3,q-1]$ MDS codes over 
$\gf(q)$, which are called oval codes. In this paper, we investigate the binary subfield codes of two families of hyperoval codes and the $p$-ary subfield codes of the conic codes. The weight distributions of these subfield codes and the parameters of their duals are determined. As a byproduct, we generalize one family of the binary subfield codes to the $p$-ary case and obtain its weight distribution. The codes presented in this paper are optimal or almost optimal in many cases. In addition, the parameters of these binary codes and $p$-ary codes seem new.  
\end{abstract}


\end{frontmatter}

\section{Introduction}
For a prime power $q$, let $\gf(q)$ denote the finite field with $q$ elements. Let $\wt(\bc)$ denote the Hamming weight of a vector $\bc\in \gf(q)^n$. An $[n,k,d]$ \emph{code} $\C$ over $\gf(q)$ is a $k$-dimensional subspace
of $\gf(q)^n$ with minimum Hamming distance $d$. Denote by $\C^\bot$ the dual code of a linear code $\C$. We call an $[n,k,d]$ code \emph{distance-optimal} if no $[n,k,d+1]$ code exists and \emph{dimension-optimal} if no $[n,k+1,d]$ code exists. Let $A_i$ denote the number of codewords with Hamming weight $i$ in a code $\C$ of length $n$.
The \emph{weight enumerator} of $\C$ is defined by $1+A_1z+A_2z^2+ \cdots + A_nz^n.$ The sequence $(1, A_1, A_2, \cdots, A_n)$ is referred to as the \emph{weight distribution} of the code $\C$. A code $\C$ is said to be a $t$-weight code  if the number of nonzero $A_i$ in the sequence $(A_1, A_2, \cdots, A_n)$ is equal to $t$. The weight distribution of a code is used to estimate the error correcting capability and compute the error probability of error detection and correction of the code \cite{KL}. The weight distributions of linear codes have also applications in cryptography and combinatorics.

Let $\PG(2,\gf(q))$ denote the projective plane over $\gf(q)$. An $r$-\emph{arc} $\cA$ of 
$\PG(2,\gf(q))$ is a set of $r$ points in $\PG(2,\gf(q))$ such that no three of them are 
collinear, where $r \geq 3$. It is known that $|\cA| \leq q+2$ for even $q$. It is conjectured 
that $|\cA| \leq q+1$ for odd $q$ \cite{B}.  

When $q$ is odd, $(q+1)$-arcs are called \emph{ovals}. When $q$ is even, $(q+2)$-arcs are referred to as \emph{hyperovals}. Hyperovals are maximal arcs as they have the maximal number of points as arcs. For even $q$, all hyperovals in $\PG(2,\gf(q))$ can be constructed with a special type of permutation polynomials on $\gf(q)$, which is described in the following theorem.

\begin{theorem}\label{th-hyperovals}\cite[Th. 9.67]{LN}
Let $q>2$ be a power of 2. Any hyperoval $\cH$ in $\PG(2,\gf(q))$ can be written in the form
$$\cH(f)=\{(f(c),c,1):c\in \gf(q)\}\cup \{(1,0,0)\}\cup\{(0,1,0)\},$$ where $f\in \gf(q)[x]$ is such that
\begin{enumerate}
\item[(i)] $f$ is a permutation polynomial of $\gf(q)$ with $\deg(f)<q$ and $f(0)=0$, $f(1)=1$;
\item[(ii)] for each $a\in \gf(q)$, $g_a(x)=(f(x+a)+f(a))x^{q-2}$ is also a permutation polynomial of $\gf(q)$.
\end{enumerate}
Conversely, every such set $\cH(f)$ is a hyperoval.
\end{theorem}
Polynomials satisfying Conditions (i) and (ii) in Theorem \ref{th-hyperovals} are called \emph{o-polynomials}. Let $q=2^m$. The following two o-polynomals over $\gf(q)$ are well known:
\begin{enumerate}
\item[(1)] (\cite{S57}) $f_1(x)=x^2$ (the corresponding oval is called the translation oval);
\item[(2)] (\cite{S62, S71},\ Segre o-polynomial) $f_2(x)=x^6$ with $m$  (the Segre oval).
\end{enumerate}
More constructions of o-polynomials could be found in \cite{DY}.

Given a hyperoval $\cH=\{\bh_1,\bh_2,\cdots,\bh_{q+2}\}$ in $\PG(2,\gf(q))$ with $q$ even, we construct a linear code $\C_{\cH}$ of length $q+2$ over $\gf(q)$ with generator matrix
$$G_{\cH}=[\bh_1\bh_2\cdots\bh_{q+2}]$$
where each $\bh_i$ is a column vector in $\gf(q)^3$.  
Since a hyperoval in $\PG(2,\gf(q))$ meets each line either in 0 or 2 points, the code $\C_{\cH}$ has only the nonzero weights $q$ and $q+2$ and $\C_{\cH}$ is projective.  Then the parameters and the weight enumerator of $\C_{\cH}$ are given in the following theorem.
\begin{theorem}\label{th-hyperovalcode}
The hyperoval code $\C_{\cH}$ is a two-weight MDS $[q+2,3,q]$ code over $\gf(q)$ and has weight enumerator
$$1+\frac{(q+2)(q^2-1)}{2}z^q+\frac{q(q-1)^2}{2}z^{q+2}.$$
The dual $\C_{\cH}^{\perp}$ of $\C_{\cH}$ is an MDS $[q+2,q-1,4]$ code over $\gf(q)$.
\end{theorem}
Conversely, the column vectors of a generator matrix of any MDS $[q+2,3,q]$ code over $\gf(q)$ form a hyperoval in $\PG(2,\gf(q))$. Thus constructing hyperovals in $\PG(2,\gf(q))$ is equivalent to constructing $[q+2,3,q]$ codes over $\gf(q)$. 
Therefore, every $[q+2,3,q]$ code over $\gf(q)$ is called a \emph{hyperoval code.}

A \emph{conic} in $\PG(2,\gf(q))$ is a set of $q+1$ points of $\PG(2,\gf(q))$ that are zeros of a nondegenerate homogeneous quadratic form in three variables. It is known that a  conic is an oval in $\PG(2,\gf(q))$ and an oval in $\PG(2,\gf(q))$ is a conic if $q$ is odd \cite{AK}.  
Hence, conics and ovals in  $\PG(2,\gf(q))$ are the same when $q$ is odd. Let $q$ be odd. Define 
$$\cO=\{(x^2,x,1):x\in \gf(q)\}\cup \{(1,0,0)\}.$$ 
It is well known that $\cO$ is a conic. 

We construct a linear code $\C_{\cO}$ of length $q+1$ over $\gf(q)$ with generator matrix
$$G_{\cO}=\begin{bmatrix} x_1^2 & x_2^2 & \cdots & x_q^2 & 1\\
 x_1 & x_2 & \cdots & x_q & 0\\
 1 & 1 & \cdots & 1 & 0\\ \end{bmatrix},$$ where $\gf(q)=\{x_1,x_2,\cdots,x_{q}\}$. 
Note that any line in $\PG(2, \gf(q))$ meets an oval in at most two points. In addition, 
it is known that $\cO$ has external lines, tangent lines and secants. The following 
theorem then follows, and is known in the literature. 

\begin{theorem}
The conic code 
$\C_{\cO}$ has parameters $[q+1,3,q-1]$ and weight enumerator 
$$1+\frac{q(q^2-1)}{2}z^{q-1}+(q^2-1)z^q+\frac{q(q-1)^2}{2}z^{q+1}.$$ 
Its dual is also an MDS code with parameters $[q+1,q-2,4]$. 
\end{theorem}

Let $\gf(q^m)$ be a finite field with $q^m$ elements, where $q$ is a power of a prime and $m$ is a positive integer. Given an $[n,k]$ code $\C$ over $\gf(q^m)$, we construct a new $[n, k']$ code $\C^{(q)}$ over $\gf(q)$ as follows. Let $G$ be a generator matrix of $\C$. Take a basis of $\gf(q^m)$ over $\gf(q)$. Represent each entry of $G$ as an $m \times 1$ column vector of $\gf(q)^m$ with respect to this basis, and replace each entry of $G$ with the corresponding $m \times 1$ column vector of $\gf(q)^m$. In this way, $G$ is modified into a $km \times n$ matrix over $\gf(q)$, which generates the new \emph{subfield code} $\C^{(q)}$ over $\gf(q)$ with length $n$.
It is known that the subfield code $\C^{(q)}$ is independent of both the choice of the basis of $\gf(q^m)$ over $\gf(q)$ and the choice of the generator matrix $G$ of $\C$ (see Theorems 2.1 and 2.6 in \cite{DH}). Note that subfield codes are different from the subfield subcodes well studied in the literature.

For a linear code $\C$ over $\gf(q^m)$, a relationship between the minimal distance of $\C^{\bot}$  and that of $\C^{(q)\bot}$ is given as follows.
\begin{lemma}\label{th-dualdistance}\cite[Th 2.7]{DH}
The minimal distance $d^\perp$ of $\C^{\bot}$ and the minimal distance $d^{(q)\perp}$ of $\C^{(q)\bot}$ satisfy
$$d^{(q)\perp}\geq d^\perp.$$
\end{lemma}

By definition, the dimension $k'$ of $\C^{(q)}$ satisfies $k'\leq mk$. To the best of our knowledge, the only references 
on subfield codes are \cite{CCD, CCZ, DH}. Recently, some basic 
results about subfield codes were derived and the subfield codes of ovoid codes were studied 
in \cite{DH}. It was demonstrated that the subfield codes of ovoid codes are very attractive 
\cite{DH}.

The first objective of this paper is to investigate the binary subfiled codes $\C_{\cH(f_1)}^{(2)}$ and $\C_{\cH(f_2)}^{(2)}$ of the hyperoval codes $\C_{\cH(f_1)}$ and $\C_{\cH(f_2)}$, respectively. The weight distributions of $\C_{\cH(f_1)}^{(2)}$ and $\C_{\cH(f_2)}^{(2)}$ are determined. The parameters of the duals of $\C_{\cH(f_1)}^{(2)}$ and $\C_{\cH(f_2)}^{(2)}$ are also obtained. As a byproduct, we generalize the binary subfield code $\C_{\cH(f_1)}^{(2)}$ to the $p$-ary case and derive its weight distribution for odd $p$. The second objective of this paper is to study the subfield codes $\C_{\cO}^{(p)}$ of the conic codes $\C_{\cO}$ for odd $p$. The weight distribution of the subfield code $\C_{\cO}^{(p)}$ is also determined. The codes presented in this paper are optimal or almost optimal in many cases. The parameters of the linear codes presented in this paper are new. 

\section{Preliminaries}\label{sect-pre}

In this section, we recall characters and some character sums over finite fields which will be needed later. 

Let $p$ be a prime and $q=p^m$. Let $\gf(q)$ be the finite field with $q$ elements and $\alpha$ a primitive element of $\gf(q)$. Let $\tr_{q/p}$ denote the trace function from $\gf(q)$ to $\gf(p)$ given by
$$\tr_{q/p}(x)=\sum_{i=0}^{m-1}x^{p^{i}},\ x\in \gf(q).$$ Denote by $\zeta_p$ the primitive $p$-th root of complex unity.

An \emph{additive character} of $\gf(q)$ is a function $\chi: (\gf(q),+)\rightarrow \bC^{*}$ such that
$$\chi(x+y)=\chi(x)\chi(y),\ x,y\in \gf(q),$$ where $\bC^{*}$ denotes the set of all nonzero complex numbers. For any $a\in \gf(q)$, the function
$$\chi_{a}(x)=\zeta_{p}^{\tr_{q/p}(ax)},\ x\in \gf(q),$$ defines an additive character of $\gf(q)$. In addition, $\{\chi_{a}:a\in \gf(q)\}$ is a group consisting of all the additive characters of $\gf(q)$. If $a=0$, we have $\chi_0(x)=1$ for all $x\in \gf(q)$ and $\chi_0$ is referred to as the trivial additive character of $\gf(q)$. If $a=1$, we call $\chi_1$ the canonical additive character of $\gf(q)$. Clearly, $\chi_a(x)=\chi_1(ax)$. The orthogonality  relation of additive characters is given by
$$\sum_{x\in \gf(q)}\chi_1(ax)=\left\{
\begin{array}{rl}
q    &   \mbox{ for }a=0,\\
0    &   \mbox{ for }a\in \gf(q)^*.
\end{array} \right. $$

Let $\gf(q)^*=\gf(q)\setminus \{0\}$. A \emph{character} $\psi$ of the multiplicative group $\gf(q)^*$ is a function from  $\gf(q)^*$  to $\bC^{*}$ such that $\psi(xy)=\psi(x)\psi(y)$ for all $(x,y)\in \gf(q)^*\times \gf(q)^*$. Define the multiplication of two characters $\psi,\psi'$ by $(\psi\psi')(x)=\psi(x)\psi'(x)$ for $x\in \gf(q)^*$. All the characters of $\gf(q)^*$ are given by
$$\psi_{j}(\alpha^k)=\zeta_{q-1}^{jk}\mbox{ for }k=0,1,\cdots,q-1,$$
where $0\leq j \leq q-2$. Then all these $\psi_j$, $0\leq j \leq q-2$, form a group under the multiplication of characters and are called \emph{multiplicative characters} of $\gf(q)$. In particular, $\psi_0$ is called the trivial multiplicative character and $\eta:=\psi_{(q-1)/2}$ is referred to as the quadratic multiplicative character of  $\gf(q)$. The orthogonality relation of multiplicative characters is given by
$$\sum_{x\in \gf(q)^*}\psi_j(x)=\left\{
\begin{array}{rl}
q-1    &   \mbox{ for }j=0,\\
0    &   \mbox{ for }j\neq 0.
\end{array} \right. $$

For an additive character $\chi$ and a multiplicative character $\psi$ of $\gf(q)$, the \emph{Gauss sum} $G(\psi, \chi)$ over $\gf(q)$ is defined by
$$G(\psi,\chi)=\sum_{x\in \gf(q)^*}\psi(x)\chi(x).$$
We call $G(\eta,\chi)$ the quadratic Gauss sum over $\gf(q)$ for nontrivial $\chi$. The value of the quadratic Gauss sum is known and documented below.

\begin{lemma}\label{quadGuasssum}\cite[Th. 5.15]{LN}
Let $q=p^m$ with $p$ odd. Let $\chi$ be the canonical additive character of $\gf(q)$. Then
\begin{eqnarray*}G(\eta,\chi)&=&(-1)^{m-1}(\sqrt{-1})^{(\frac{p-1}{2})^2m}\sqrt{q}\\
 &=&\left\{
\begin{array}{lll}
(-1)^{m-1}\sqrt{q}    &   \mbox{ for }p\equiv 1\pmod{4},\\
(-1)^{m-1}(\sqrt{-1})^{m}\sqrt{q}    &   \mbox{ for }p\equiv 3\pmod{4}.
\end{array} \right. \end{eqnarray*}
\end{lemma}

Let $\chi$ be a nontrivial additive character of $\gf(q)$ and let $f\in \gf(q)[x]$ be a polynomial of positive degree. The character sums of the form
$$\sum_{c\in \gf(q)}\chi(f(c))$$ are referred to as \emph{Weil sums}. The problem of evaluating
such character sums explicitly is very difficult in general. In certain special cases, Weil sums can be treated (see \cite[Section 4 in Chapter 5]{LN}).

If $f$ is a quadratic polynomial and $q$ is odd, the Weil sums have an interesting relationship with quadratic Gauss sums, which is described in the
following lemma.

\begin{lemma}\label{lem-charactersum}\cite[Th. 5.33]{LN}
Let $\chi$ be a nontrivial additive character of $\gf(q)$ with $q$ odd, and let $f(x)=a_2x^2+a_1x+a_0\in \gf(q)[x]$ with $a_2\neq 0$. Then
$$\sum_{c\in \gf(q)}\chi(f(c))=\chi(a_0-a_1^2(4a_2)^{-1})\eta(a_2)G(\eta,\chi).$$
\end{lemma}

If $f$ is a quadratic polynomial with $q$ even, the Weil sums are evaluated explicitly as follows.

\begin{lemma}\label{lem-charactersum-evenq}\cite[Cor. 5.35]{LN}
Let $\chi_b$ be a nontrivial additive character of $\gf(q)$ with $b\in \gf(q)^*$, and let $f(x)=a_2x^2+a_1x+a_0\in \gf(q)[x]$ with $q$ even. Then
$$\sum_{c\in \gf(q)}\chi_b(f(c))=\left\{\begin{array}{ll}
\chi_b(a_0)q    &   \mbox{ if }a_2=ba_{1}^{2},\\
0    &   \mbox{ otherwise. }
\end{array} \right.$$
\end{lemma}

The Weil sums can also be evaluated explicitly in the case that $f$ is an affine $p$-polynomial over $\gf(q)$.

\begin{lemma}\label{lem-p-polynomial}\cite[Th. 5.34]{LN}
Let $q=p^m$ and let
$$f(x)=a_rx^{p^r}+a_{r-1}x^{p^{r-1}}+\cdots+a_1x^{p}+a_0x+a$$ be an affine $p$-polynomial over $\gf(q)$. Let $\chi_b$ be a nontrivial additive character of $\gf(q)$ with $b\in \gf(q)^*$. Then
$$\sum_{c\in \gf(q)}\chi_b(f(c))=\left\{\begin{array}{ll}
\chi_b(a)q    &   \mbox{ if }ba_r+b^pa_{r-1}^p+\cdots+b^{p^{r-1}}a_{1}^{p^{r-1}}+b^{p^{r}}a_{0}^{p^{r}}=0,\\
0    &   \mbox{ otherwise. }
\end{array} \right.$$
\end{lemma}

\section{The subfield codes of the translation and Segre hyperoval codes}\label{set-hyperovalcode}

Let $q=2^m$ with $m>1$. By Theorem \ref{th-hyperovals},  any hyperoval $\cH$ in $\PG(2,\gf(q))$ can be written in the form
$$\cH(f)=\{(f(c),c,1):c\in \gf(q)\}\cup \{(1,0,0)\}\cup\{(0,1,0)\},$$ where $f\in \gf(q)[x]$ is an o-polynomial. Let $\gf(q)=\{x_1,x_2,\cdots,x_{q}\}$. For a hyperoval $\cH(f)$ in $\PG(2,\gf(q))$, we denote
$$G_{\cH(f)}=\begin{bmatrix} f(x_1) & f(x_2) & \cdots & f(x_{q}) & 1 & 0\\
 x_1 & x_2 & \cdots & x_q & 0 & 1\\
 1 & 1 & \cdots & 1 & 0 & 0\\ \end{bmatrix}.$$
Let $\C_{\cH(f)}$ be the $[q+2,3]$ code with generator matrix $G_{\cH(f)}$. Denote by $\C_{\cH(f)}^{\perp}$ and $\C_{\cH(f)}^{(2)\perp}$ the dual codes of $\C_{\cH(f)}$ and its subfield code $\C_{\cH(f)}^{(2}$, respectively.

To give the trace representation of $\C_{\cH(f)}^{(2)}$, we recall the following result.
\begin{lemma}\label{th-tracerepresentation}\cite[Th. 2.5]{DH}
Let $\C$ be an $[n,k]$ code over $\gf(q^m)$. Let $G=[g_{ij}]_{1\leq i \leq k, 1\leq j \leq n}$ be a generator matrix of $\C$. Then the trace representation of the subfield code $\C^{(q)}$ is given by
$$
\C^{(q)}=\left\{\left(\tr_{q^m/q}\left(\sum_{i=1}^{k}a_ig_{i1}\right),
\cdots,\tr_{q^m/q}\left(\sum_{i=1}^{k}a_ig_{in}\right)\right):a_1,\ldots,a_k\in \gf(q^m)\right\}.
$$
\end{lemma}

Lemma \ref{th-tracerepresentation} directly gives the following trace representation of $\C_{\cH(f)}^{(2)}$: 
\begin{eqnarray}\label{eqn-tracerepresentation}
\C_{\cH(f)}^{(2)}=\left\{\left((\tr_{q/2}(af(x)+bx)+c)_{x\in \gf(q)},\tr_{q/2}(a),\tr_{q/2}(b)\right):a,b\in \gf(q),c\in \gf(2)\right\}.
\end{eqnarray}

Let $f_1(x)=x^2\in\gf(q)[x]$.  Let $f_2(x)=x^6\in\gf(q)[x]$ for odd $m$. Then $f_1$ 
and $f_2$ are the translation and Segre o-polynomials and  $\C_{\cH(f_1)}$ and 
$\C_{\cH(f_2)}$ are the corresponding hyperoval codes. In the following, we investigate 
the subfield codes of $\C_{\cH(f_1)}$ and $\C_{\cH(f_2)}$.

\subsection{The subfield code of $\C_{\cH(f_1)}$}\label{subsection-main1}

By Equation (\ref{eqn-tracerepresentation}), the trace representation of $\C_{\cH(f_1)}^{(2)}$ is
\begin{eqnarray*}
\C_{\cH(f_1)}^{(2)}=\left\{\left((\tr_{q/2}(ax^2+bx)+c)_{x\in \gf(q)},\tr_{q/2}(a),\tr_{q/2}(b)\right):a,b\in \gf(q),c\in \gf(2)\right\}.
\end{eqnarray*}

\begin{theorem}\label{mian-1}
Let $m>1$ and $q=2^m$. Then $\C_{\cH(f_1)}^{(2)}$ is a $[2^m+2,m+2,2]$ binary linear code with the weight distribution in Table \ref{tab-1}. Its dual has parameters $[2^m+2,2^m-m,4]$ and 
is distance-optimal according to the sphere-packing bound.
\end{theorem}

\begin{table}[ht]
\begin{center}
\caption{The weight distribution of $\C_{\cH(f_1)}^{(2)}$}\label{tab-1}
\begin{tabular}{cc} \hline
Weight  &  Multiplicity   \\ \hline
$0$          &  $1$ \\
$2$  &  $1$ \\
$2^{m}$  & $1$ \\
$2^{m}+2$    & $1$ \\
$2^{m-1}$    & $2(2^{m-1}-1)$ \\
$2^{m-1}+1$    & $2^{m+1}$ \\
$2^{m-1}+2$    & $2(2^{m-1}-1)$ \\
\hline
\end{tabular}
\end{center}
\end{table}

\begin{proof}
Let $\chi$ be the canonical additive character of $\gf(q)$.
Denote
$$N_0(a,b)=\sharp \{x\in \gf(q):\tr_{q/2}(ax^2+bx)=0\}.$$ By the orthogonality relation of additive characters and Lemma \ref{lem-charactersum-evenq}, we have
\begin{eqnarray*}\label{eqn-1}
\nonumber 2N_0(a,b)&=&\sum_{z\in \gf(2)}\sum_{x\in \gf(q)}(-1)^{z\tr_{q/2}(ax^2+bx)}\\
\nonumber&=&q+\sum_{x\in \gf(q)}\chi(ax^2+bx)\\
&=&\left\{\begin{array}{ll}
2q   &   \mbox{ if }a=b^{2},\\
q    &   \mbox{ otherwise. }\\
\end{array} \right.
\end{eqnarray*}
Note that $\tr_{q/2}(b^2)=\tr_{q/2}(b)$. For any codeword $$\bc(a,b,c)=\left((\tr_{q/2}(ax^2+bx)+c)_{x\in \gf(q)},\tr_{q/2}(a),\tr_{q/2}(b)\right)\in \C_{\cH(f_1)}^{(2)},$$ we discuss its Hamming weight in the following two cases.
\begin{enumerate}
\item[$\bullet$] If $c=0$, then we have
\begin{eqnarray*}
\wt(\bc(a,b,c))&=&
\left\{\begin{array}{ll}
q-N_0(a,b)   &   \mbox{ for }a=b^{2},\ \tr_{q/2}(a)=\tr_{q/2}(b)=0 \\
q-N_0(a,b)+2    &   \mbox{ for }a=b^{2},\ \tr_{q/2}(a)=\tr_{q/2}(b)\neq0 \\
q-N_0(a,b)   &   \mbox{ for }a\neq b^{2},\ \tr_{q/2}(a)=\tr_{q/2}(b)=0 \\
q-N_0(a,b)+1 & \myatop{\mbox{ for $a\neq b^{2},\ \tr_{q/2}(a)=0,\ \tr_{q/2}(b)\neq0,$}}{\mbox{ or $a\neq b^{2},\ \tr_{q/2}(a)\neq 0,\ \tr_{q/2}(b)=0$}} \\
q-N_0(a,b)+2    &   \mbox{ for }a\neq b^{2},\ \tr_{q/2}(a)\neq 0,\ \tr_{q/2}(b)\neq0 \\
\end{array} \right.\\
&=&\left\{\begin{array}{ll}
0 &\mbox{ for }a=b^{2},\ \tr_{q/2}(a)=\tr_{q/2}(b)=0,\\
2 &\mbox{ for }a=b^{2},\ \tr_{q/2}(a)=\tr_{q/2}(b)\neq0,\\
2^{m-1}  &\mbox{ for }a\neq b^{2},\ \tr_{q/2}(a)=\tr_{q/2}(b)=0,\\
2^{m-1}+1 &\myatop{\mbox{ for $a\neq b^{2},\ \tr_{q/2}(a)=0,\ \tr_{q/2}(b)\neq0,$}}{\mbox{ or $a\neq b^{2},\ \tr_{q/2}(a)\neq 0,\ \tr_{q/2}(b)=0,$}}\\
2^{m-1}+2 &\mbox{ for }a\neq b^{2},\ \tr_{q/2}(a)\neq 0,\ \tr_{q/2}(b)\neq0.
\end{array} \right.\\
\end{eqnarray*}
\item[$\bullet$] If $c=1$, then we have
\begin{eqnarray*}
\wt(\bc(a,b,c))&=&
\left\{\begin{array}{ll}
N_0(a,b)   &   \mbox{ for }a=b^{2},\ \tr_{q/2}(a)=\tr_{q/2}(b)=0 \\
N_0(a,b)+2    &   \mbox{ for }a=b^{2},\ \tr_{q/2}(a)=\tr_{q/2}(b)\neq0 \\
N_0(a,b)   &   \mbox{ for }a\neq b^{2},\ \tr_{q/2}(a)=\tr_{q/2}(b)=0 \\
N_0(a,b)+1 & \myatop{\mbox{ for $a\neq b^{2},\ \tr_{q/2}(a)=0,\ \tr_{q/2}(b)\neq0,$}}{\mbox{ or $a\neq b^{2},\ \tr_{q/2}(a)\neq 0,\ \tr_{q/2}(b)=0$}}\\
N_0(a,b)+2    &   \mbox{ for }a\neq b^{2},\ \tr_{q/2}(a)\neq 0,\ \tr_{q/2}(b)\neq0 \\
\end{array} \right.\\
&=&\left\{\begin{array}{ll}
2^m &\mbox{ for }a=b^{2},\ \tr_{q/2}(a)=\tr_{q/2}(b)=0,\\
2^m+2 &\mbox{ for }a=b^{2},\ \tr_{q/2}(a)=\tr_{q/2}(b)\neq0,\\
2^{m-1}  &\mbox{ for }a\neq b^{2},\ \tr_{q/2}(a)=\tr_{q/2}(b)=0,\\
2^{m-1}+1 &\myatop{\mbox{ for $a\neq b^{2},\ \tr_{q/2}(a)=0,\ \tr_{q/2}(b)\neq0,$}}{\mbox{ or $a\neq b^{2},\ \tr_{q/2}(a)\neq 0,\ \tr_{q/2}(b)=0,$}}\\
2^{m-1}+2 &\mbox{ for }a\neq b^{2},\ \tr_{q/2}(a)\neq 0,\ \tr_{q/2}(b)\neq0.
\end{array} \right.\\
\end{eqnarray*}
\end{enumerate}
Observe that the Hamming weight 0 occurs $2^{m-1}$ times if $(a,b,c)$ runs through $\gf(q)\times \gf(q)\times \gf(2)$. Thus every codeword in $\C_{\cH(f_1)}^{(2)}$ repeats $2^{m-1}$ times. Based on the discussions above, we easily deduce the weight distribution of $\C_{\cH(f_1)}$.

Note that the dual of $\C_{\cH(f_1)}^{(2)}$ has length $2^m+2$ and dimension $2^m-m$. By Theorem  \ref{th-hyperovalcode} and Lemma \ref{th-dualdistance}, the minimal distance of  $\C_{\cH(f_1)}^{(2)}$ satisfies $d^{(2)\perp}\geq 4$. By the sphere-packing bound,
$$2^{2^{m}+2}\geq 2^{2^{m}-m}\left(\sum_{i=0}^{^{\left\lfloor\frac{d^{(2)\perp}-1}{2}\right\rfloor}}\binom{2^{m}+2}{i}\right),$$
then $d^{(2)\perp}\leq 4$. Hence $d^{(2)\perp}= 4$. We complete the proof.
\end{proof}

\begin{example}\label{exa-1}
Let $m=2$. Then $\C_{\cH(f_1)}^{(2)}$ in Theorem \ref{mian-1} is a $[6,4,2]$ binary linear code and its dual has parameters $[6,2,4]$. Both of $\C_{\cH(f_1)}^{(2)}$ and its dual are optimal according to the tables of best codes known maintained at http://www.codetables.de.
\end{example}

\subsection{The subfield code of $\C_{\cH(f_2)}$}

Let $q=2^m$ with $m$ odd.
By Equation (\ref{eqn-tracerepresentation}), the trace representation of $\C_{\cH(f_2)}^{(2)}$ is
\begin{eqnarray*}
\C_{\cH(f_2)}^{(2)}=\left\{\left((\tr_{q/2}(ax^6+bx)+c)_{x\in \gf(q)},\tr_{q/2}(a),\tr_{q/2}(b)\right):a,b\in \gf(q),c\in \gf(2)\right\}.
\end{eqnarray*}

\begin{lemma}\label{lem-numberofb}
Let $q=2^m$ with $m$ odd. Let $a\in \gf(q)^*$ and let $y_a$ be the unique solution of $g(y)=ay^6=1$ in $\gf(q)^*$. Then the following statements hold.
\begin{enumerate}
\item[(1)] $$\sharp \{b\in \gf(q):\tr_{q/2}(b)=0\mbox{ and }\tr_{q/2}(by_a)=0\}=\left\{\begin{array}{ll}
2^{m-1}  & \mbox{ if }a=1,\\
2^{m-2} & \mbox{ if }a\neq1.
\end{array} \right.$$
\item[(2)]
$$\sharp \{b\in \gf(q):\tr_{q/2}(b)=1\mbox{ and }\tr_{q/2}(by_a)=0\}=\left\{\begin{array}{ll}
0  & \mbox{ if }a=1,\\
2^{m-2} & \mbox{ if }a\neq1.
\end{array} \right.$$
\item[(3)] $$\sharp \{b\in \gf(q):\tr_{q/2}(b)=0\mbox{ and }\tr_{q/2}(by_a)=1\}=\left\{\begin{array}{ll}
0  & \mbox{ if }a=1,\\
2^{m-2} & \mbox{ if }a\neq1.
\end{array} \right.$$
\item[(4)]
$$\sharp \{b\in \gf(q):\tr_{q/2}(b)=1\mbox{ and }\tr_{q/2}(by_a)=1\}=\left\{\begin{array}{ll}
2^{m-1}  & \mbox{ if }a=1,\\
2^{m-2} & \mbox{ if }a\neq1.
\end{array} \right.$$
\end{enumerate}
\end{lemma}

\begin{proof}
 Let $\chi$ be the canonical additive character of $\gf(q)$ and $g(y)=ay^6$. Then $g(y)$ is a permutation polynomial of $\gf(q)$ as $\gcd(6,2^m-1)=\gcd(2+1,2^m-1)=1$. This implies that $ay^6=1$ has a unique solution $y_a\in \gf(q)^*$. Then
 \begin{eqnarray*}
 \lefteqn{ \sharp \{b\in \gf(q):\tr_{q/2}(b)=0\mbox{ and }\tr_{q/2}(by_a)=0\} } \\
 &=&\frac{1}{4}\sum_{z_1\in \gf(2)}\sum_{z_2\in \gf(2)}\sum_{b\in \gf(q)}\chi(z_1b)\chi(z_2y_ab)\\
 &=&\frac{q}{4}+\frac{1}{4}\sum_{b\in \gf(q)}\chi((1+y_a)b)\\
 &=&\left\{\begin{array}{ll}
2^{m-1}  & \mbox{ if }y_a=1,\\
2^{m-2} & \mbox{ if }y_a\neq1.\\
\end{array} \right.\\
 \end{eqnarray*}
 Note that $y_a=1$ if and only if $a=1$. Then the first equation holds and the others follow.
\end{proof}

\begin{theorem}\label{th-main2}
Let $m>1$ be odd and $q=2^m$. Then $\C_{\cH(f_2)}^{(2)}$ is a $[2^m+2,2m+1,2^{m-1}-2^{\frac{m-1}{2}}]$ binary linear code with the weight distribution in Table \ref{tab-2}. Its dual has parameters $[2^m+2,2^m-2m+1,4]$.
\end{theorem}

\begin{table}[ht]
\begin{center}
\caption{The weight distribution of $\C_{\cH(f_2)}^{(2)}$}\label{tab-2}
\begin{tabular}{cc} \hline
Weight  &  Multiplicity   \\ \hline
0 & 1\\
 $2^m$ & 1\\
 $2^{m-1}$ & $(2^{m-1}-1)(2^{m-1}+2)$\\
 $2^{m-1}+1$ & $2^{m}(2^{m-1}+1)$\\
 $2^{m-1}+2$ & $2^{m-1}(2^{m-1}-1)$\\
 $2^{m-1}+2^{(m-1)/2}$ & $2^{m-2}(2^{m-1}-1)$\\
 $2^{m-1}-2^{(m-1)/2}$ & $2^{m-2}(2^{m-1}-1)$\\
 $2^{m-1}+2^{(m-1)/2}+1$ & $2^{m-1}(2^{m-1}-1)$\\
 $2^{m-1}-2^{(m-1)/2}+1$ & $2^{m-1}(2^{m-1}-1)$\\
 $2^{m-1}+2^{(m-1)/2}+2$ & $2^{m-2}(2^{m-1}+1)$ \\
 $2^{m-1}-2^{(m-1)/2}+2$ & $2^{m-2}(2^{m-1}+1)$\\
\hline
\end{tabular}
\end{center}
\end{table}

\begin{proof}
Let $\chi$ be the canonical additive character of $\gf(q)$.
Denote $N_0(a,b)=\sharp \{x\in \gf(q):\tr_{q/2}(ax^6+bx)=0\}$. By the orthogonality relation of additive characters, we have
\begin{eqnarray}\label{eqn-2}
\nonumber N_0(a,b)&=&\frac{1}{2}\sum_{z\in \gf(2)}\sum_{x\in \gf(q)}(-1)^{z\tr_{q/2}(ax^6+bx)}\\
&=&2^{m-1}+\frac{1}{2}\sum_{x\in \gf(q)}\chi(ax^6+bx).
\end{eqnarray}
Denote $\Delta(a,b)=\sum_{x\in \gf(q)}\chi(ax^6+bx)$. We discuss the value of $\Delta$ in the  cases below.
\begin{enumerate}
\item[(1)] Let $a=b=0$. Then $\Delta(a,b)=q$.
\item[(2)] Let $a=0,b\neq 0$. Then $\Delta(a,b)=0$.
\item[(3)] Let $a\neq 0$. By Lemma \ref{lem-p-polynomial} we have
\begin{eqnarray*}
\Delta(a,b)^2&=&\sum_{x_1\in \gf(q)}\chi(ax_{1}^{6}+bx_1)\sum_{x\in \gf(q)}\chi(ax^6+bx)\\
&=&\sum_{y\in \gf(q)}\chi(a(x+y)^{6}+b(x+y))\sum_{x\in \gf(q)}\chi(ax^6+bx)\\
&=&\sum_{x,y\in \gf(q)}\chi\left(a(x^4+y^4)(x^2+y^2)+b(x+y)+ax^6+bx\right)\\
&=&\sum_{y\in \gf(q)}\chi(ay^{6}+by)\sum_{x\in \gf(q)}\chi(ay^2x^4+ay^4x^2)\\
&=&q+q\sum_{\myatop{y\in \gf(q)^*}{ay^2(1+ay^6)=0}}\chi(ay^{6}+by)\\
&=&q+q\sum_{ay^6=1}\chi(1+by).
\end{eqnarray*}
Since $m$ is odd, we have $\chi(1)=-1$. Let $y_a$ be the unique solution of $g(y)=ay^6=1$ in $\gf(q)^*$. Thus we further have
\begin{eqnarray*}\Delta(a,b)^2&=&q-q\chi(by_a)\\
&=&\left\{\begin{array}{ll}
0  & \mbox{ if }\tr_{q/2}(by_a)=0,\\
2q & \mbox{ if }\tr_{q/2}(by_a)=1,\\
\end{array} \right.\\
\end{eqnarray*}
Hence
\begin{eqnarray*}\Delta(a,b)
=\left\{\begin{array}{ll}
0  & \mbox{ if }\tr_{q/2}(by_a)=0,\\
\pm 2^{(m+1)/2} & \mbox{ if }\tr_{q/2}(by_a)=1.\\
\end{array} \right.
\end{eqnarray*}
\end{enumerate}
By Equation (\ref{eqn-2}) and the discussions above, we have
\begin{eqnarray}\label{eqn-3}
 N_0(a,b)=\left\{\begin{array}{ll}
 2^m & \mbox{ if }a=b=0,\\
 2^{m-1} & \myatop{\mbox{if $\tr_{q/2}(by_a)=0,a\neq 0,$}}{\mbox{or $a=0,b\neq 0,$}}\\
2^{m-1}\pm 2^{(m-1)/2} & \mbox{ if }\tr_{q/2}(by_a)=1,a\neq 0,\\
\end{array} \right.
\end{eqnarray}
where $ay_a^6=1$.

For any codeword $\bc(a,b,c)=\left((\tr_{q/2}(ax^6+bx)+c)_{x\in \gf(q)},\tr_{q/2}(a),\tr_{q/2}(b)\right)\in \C_{\cH(f_2)}^{(2)}$, by Equation (\ref{eqn-3}) we deduce that
\begin{eqnarray*}
\lefteqn{ \wt(\bc(a,b,c)) } \\
&=&\left\{\begin{array}{ll}
 0 & \mbox{if }a=b=c=0 \\
   2^{m} & \mbox{if }a=b=0,\ c=1 \\
 2^{m-1} & \mbox{if }a=0,\ c\in \gf(2),\ b\neq 0,\ \tr_{q/2}(b)=0 \\
 2^{m-1}+1 & \mbox{if }a=0,\ c\in \gf(2),\ b\neq 0,\ \tr_{q/2}(b)\neq0 \\
 2^{m-1} & \mbox{if $\tr_{q/2}(by_a)=0,\ a\neq 0,\ \tr_{q/2}(a)=\tr_{q/2}(b)=0$},\ c\in \gf(2)\\
 2^{m-1}+1 & \mbox{if $\tr_{q/2}(by_a)=0,\ a\neq 0,\ \tr_{q/2}(a)=1,\ \tr_{q/2}(b)=0$},\ c\in \gf(2)\\
 2^{m-1}+1 & \mbox{if $\tr_{q/2}(by_a)=0,\ a\neq 0,\ \tr_{q/2}(a)=0,\ \tr_{q/2}(b)=1$},\ c\in \gf(2)\\
 2^{m-1}+2 & \mbox{if $\tr_{q/2}(by_a)=0,\ a\neq 0,\ \tr_{q/2}(a)=1,\ \tr_{q/2}(b)=1$},\ c\in \gf(2)\\
2^{m-1}\pm 2^{(m-1)/2} & \mbox{if }\tr_{q/2}(by_a)=1,\ a\neq 0,\ \tr_{q/2}(a)=\tr_{q/2}(b)=0,\ c=0\\
2^{m-1}\pm 2^{(m-1)/2}+1 & \mbox{if }\tr_{q/2}(by_a)=1,\ a\neq 0,\ \tr_{q/2}(a)=1,\ \tr_{q/2}(b)=0,\ c=0 \\
2^{m-1}\pm 2^{(m-1)/2}+1 & \mbox{if }\tr_{q/2}(by_a)=1,\ a\neq 0,\ \tr_{q/2}(a)=0,\ \tr_{q/2}(b)=1,\ c=0 \\
2^{m-1}\pm 2^{(m-1)/2}+2 & \mbox{if }\tr_{q/2}(by_a)=1,\ a\neq 0,\ \tr_{q/2}(a)=1,\ \tr_{q/2}(b)=1,\ c=0 \\
2^{m-1}\mp 2^{(m-1)/2} & \mbox{if }\tr_{q/2}(by_a)=1,\ a\neq 0,\ \tr_{q/2}(a)=\tr_{q/2}(b)=0,\ c=1 \\
2^{m-1}\mp 2^{(m-1)/2}+1 & \mbox{if }\tr_{q/2}(by_a)=1,\ a\neq 0,\ \tr_{q/2}(a)=1,\ \tr_{q/2}(b)=0,\ c=1 \\
2^{m-1}\mp 2^{(m-1)/2}+1 & \mbox{if }\tr_{q/2}(by_a)=1,\ a\neq 0,\ \tr_{q/2}(a)=0,\ \tr_{q/2}(b)=1,\ c=1 \\
2^{m-1}\mp 2^{(m-1)/2}+2 & \mbox{if }\tr_{q/2}(by_a)=1,\ a\neq 0,\ \tr_{q/2}(a)=1,\ \tr_{q/2}(b)=1,\ c=1 \\
\end{array} \right.\\
&=&\left\{\begin{array}{ll}
 0 & \mbox{with 1 time},\\
 2^m & \mbox{with 1 time},\\
 2^{m-1} & \mbox{with $(2^{m-1}-1)(2^{m-1}+2)$ times},\\
 2^{m-1}+1 & \mbox{with $2^{m}(2^{m-1}+1)$ times},\\
 2^{m-1}+2 & \mbox{with $2^{m-1}(2^{m-1}-1)$ times},\\
 2^{m-1}+2^{(m-1)/2} & \mbox{with $2^{m-2}(2^{m-1}-1)$ times},\\
 2^{m-1}-2^{(m-1)/2} & \mbox{with $2^{m-2}(2^{m-1}-1)$ times},\\
 2^{m-1}+2^{(m-1)/2}+1 & \mbox{with $2^{m-1}(2^{m-1}-1)$ times},\\
 2^{m-1}-2^{(m-1)/2}+1 & \mbox{with $2^{m-1}(2^{m-1}-1)$ times},\\
 2^{m-1}+2^{(m-1)/2}+2 & \mbox{with $2^{m-2}(2^{m-1}+1)$ times},\\
 2^{m-1}-2^{(m-1)/2}+2 & \mbox{with $2^{m-2}(2^{m-1}+1)$ times},
\end{array} \right.\\
\end{eqnarray*}
where the frequency is easy to derive based on Lemma \ref{lem-numberofb}. We remark that the frequency of the weight $2^{m-1}+2^{(m-1)/2}+i$ and the frequency of the weight $2^{m-1}-2^{(m-1)/2}+i$ are equal for any fixed $i\in \{0,1,2\}$ because of
$(1,1,\ldots,1,0,0)\in \C_{\cH(f_2)}^{(2)}$. The dimension of $\C_{\cH(f_2)}^{(2)}$ is $2m+1$ as $A_0=1$.

Note that the dual of $\C_{\cH(f_2)}^{(2)}$ has length $2^m+2$ and dimension $2^m-2m+1$. By Theorem  \ref{th-hyperovalcode} and Lemma \ref{th-dualdistance}, the minimal distance of $\C_{\cH(f_2)}^{(2)}$ satisfy $d^{(2)\perp}\geq 4$. By the first five Pless power moments in \cite{HP}, one can prove that $d^{(2)\perp}=4$.
Then the proof is completed.
\end{proof}

\begin{example}\label{exa-2}
Let $\C_{\cH(f_2)}^{(2)}$ be the linear code in Theorem \ref{th-main2}.
\begin{enumerate}
\item[(1)] Let $m=3$. Then
      $\C_{\cH(f_2)}^{(2)}$ has parameters $[10, 7, 2]$ and its dual has parameters $[10, 3, 4]$.
\item[(2)] Let $m=5$. Then
      $\C_{\cH(f_2)}^{(2)}$ has parameters $[34, 11, 12]$ and its dual has parameters $[34, 23, 4]$.
\item[(3)] Let $m=7$. Then
      $\C_{\cH(f_2)}^{(2)}$ has parameters $[130, 15, 56]$ and its dual has parameters $[130, 115, 4]$.
\end{enumerate}
The code $\C_{\cH(f_2)}^{(2)}$ is optimal in these cases according to the tables of best codes known
maintained at http://www.codetables.de.
\end{example}

We remark that the linear code $\C_{\cH(f_2)}^{(2)}$ is not a hyperoval code if $m$ is even as $f_2(x)=x^6$ is not an o-polynomial in this case. For even $m$, the weight distribution of $\C_{\cH(f_2)}^{(2)}$ becomes very complicated, according to our Magma experiments.

\section{The weight distribution of a class of $p$-ary linear codes}\label{sect-p-ary-case}

Let $q=p^m$ and let $\tr_{q/p}$ be the trace function from $\gf(q)$ to $\gf(p)$. In this section, we generalize the subfield code $\C_{\cH(f_1)}^{(2)}$ in Section \ref{subsection-main1} to the $p$-ary linear code $\C_{\cH(f_1)}^{(p)}$, where
\begin{eqnarray*}
\C_{\cH(f_1)}^{(p)}=\left\{\left((\tr_{q/p}(ax^2+bx)+c)_{x\in \gf(q)},\tr_{q/p}(a),\tr_{q/p}(b)\right):a,b\in \gf(q),c\in \gf(p)\right\}.
\end{eqnarray*}
Our goal is to determine the weight distribution of the $p$-ary linear code $\C_{\cH(f_1)}^{(p)}$ with $p$ odd. The analysis of the code $\C_{\cH(f_1)}^{(p)}$ will be employed to obtain the 
weight distribution of the subfield code of a conic code in Section \ref{sec-coniccode}.

\begin{lemma}\label{lem-equalities}
Let $q=p^m$ with $p$ an odd prime. Then the following statements hold.
\begin{enumerate}
\item[(1)] \begin{eqnarray*}
\lefteqn{\sharp \{a\in \gf(q)^*:\eta(a)=1\mbox{ and }\tr_{q/p}(a)=0\} } \\
&=&\left\{
\begin{array}{ll}
\frac{p^{m-1}-1-(p-1)p^{\frac{m-2}{2}}(\sqrt{-1})^{\frac{(p-1)m}{2}}}{2} & \mbox{ for even $m$,} \\
\frac{p^{m-1}-1}{2} & \mbox{ for odd $m$.}
\end{array}\right.
\end{eqnarray*}
\item[(2)] \begin{eqnarray*}
\lefteqn{ \sharp \{a\in \gf(q)^*:\eta(a)=1\mbox{ and }\tr_{q/p}(a) \neq 0\} }\\
&=&\left\{
\begin{array}{ll}
\frac{(p-1)(p^{m-1}+p^{\frac{m-2}{2}}(\sqrt{-1})^{\frac{(p-1)m}{2}})}{2} & \mbox{ for even $m$,}\\
\frac{p^{m-1}(p-1)}{2} & \mbox{ for odd $m$.}
\end{array}\right.
\end{eqnarray*}
\item[(3)] \begin{eqnarray*}
\lefteqn{ \sharp \{a\in \gf(q)^*:\eta(a)=-1\mbox{ and }\tr_{q/p}(a)=0\} } \\
&=&\left\{
\begin{array}{ll}
\frac{p^{m-1}-1+(p-1)p^{\frac{m-2}{2}}(\sqrt{-1})^{\frac{(p-1)m}{2}}}{2} & \mbox{ for even $m$,}\\
\frac{p^{m-1}-1}{2} & \mbox{ for odd $m$.}
\end{array}\right.
\end{eqnarray*}
\item[(4)] \begin{eqnarray*}
\lefteqn{ \sharp \{a\in \gf(q)^*:\eta(a)=-1\mbox{ and }\tr_{q/p}(a)\neq0\} } \\
&=&\left\{
\begin{array}{ll}
\frac{(p-1)(p^{m-1}-p^{\frac{m-2}{2}}(\sqrt{-1})^{\frac{(p-1)m}{2}})}{2} & \mbox{ for even $m$,}\\
\frac{p^{m-1}(p-1)}{2} & \mbox{ for odd $m$.}
\end{array}\right.
\end{eqnarray*}
\end{enumerate}
\end{lemma}

\begin{proof}
We only prove the first equality as the others follow directly.
Let $\chi$ be the canonical additive character and $\alpha$ a primitive element of $\gf(q)$. Let $C_0$ be the cyclic group generated by $\alpha^2$.
Denote $N(a)=\sharp \{a\in \gf(q)^*:\eta(a)=1\mbox{ and }\tr_{q/p}(a)=0\}$. By the orthogonality  relation of additive characters and Lemmas \ref{quadGuasssum} and \ref{lem-charactersum}, we obtain that
\begin{eqnarray*}
N(a)&=&\frac{1}{p}\sum_{z\in \gf(p)}\sum_{a\in C_0}\chi(za)\\
&=&\frac{1}{2p}\sum_{z\in \gf(p)}\sum_{a\in \gf(q)^*}\chi(za^2)\\
&=&\frac{q-p}{2p}+\frac{1}{2p}\sum_{z\in \gf(p)^*}\sum_{a\in \gf(q)}\chi(za^2)\\
&=&\frac{q-p}{2p}+\frac{1}{2p}G(\eta,\chi)\sum_{z\in \gf(p)^*}\eta(z)\\
&=&\left\{
\begin{array}{ll}
\frac{q-p}{2p}+\frac{p-1}{2p}G(\eta,\chi) & \mbox{ for even $m$} \\
\frac{q-p}{2p} & \mbox{ for odd $m$}
\end{array}\right.\\
&=&\left\{
\begin{array}{ll}
\frac{p^{m-1}-1-(p-1)p^{\frac{m-2}{2}}(\sqrt{-1})^{\frac{(p-1)m}{2}}}{2} & \mbox{ for even $m$},\\
\frac{p^{m-1}-1}{2} & \mbox{ for odd $m$},
\end{array}\right.
\end{eqnarray*}
where the fifth equality holds due to the orthogonality relation of multiplicative characters.
\end{proof}

\begin{theorem}\label{th-main3}
Let $q=p^m$ with $p$ odd and $m$ a positive integer. The following statements hold.
\begin{enumerate}
\item[(1)] If $m=1$,  then $\C_{\cH(f_1)}^{(p)}$ is an almost MDS $[p+2,3,p-1]$ code. If $m$ is odd and $m>1$, then $\C_{\cH(f_1)}^{(p)}$ is a $[p^m+2,2m+1,p^{m-1}(p-1)-p^{\frac{m-1}{2}}]$ code. For odd $m$, the weight distribution of $\C_{\cH(f_1)}^{(p)}$ is given in Table \ref{tab-3}.
\item[(2)] If $m \geq 2$ is even, then $\C_{\cH(f_1)}^{(p)}$ is a $[p^m+2,2m+1,p^{m-1}(p-1)-(p-1)p^{\frac{m-2}{2}}]$ code with the weight distribution in Table \ref{tab-4}.
\end{enumerate}
\end{theorem}

\begin{table}[ht]
\begin{center}
\caption{The weight distribution of $\C_{\cH(f_1)}^{(p)}$ with $m$ odd}\label{tab-3}
\begin{tabular}{cc} \hline
Weight  &  Multiplicity   \\ \hline
0 & 1 \\
$p^m$ & $p-1$\\
$p^{m-1}(p-1)$ & $(p^{m-1}-1)(p+p^{m-1})$\\
$p^{m-1}(p-1)+1$ & $(p^{m}-p^{m-1})(p+2p^{m-1}-1)$\\
$p^{m-1}(p-1)+2$ & $(p^{m}-p^{m-1})^2$\\
$p^{m-1}(p-1)-p^{\frac{m-1}{2}}(-1)^{\frac{(p-1)(m+1)}{4}}$ & $\frac{p^{m-1}(p-1)(p^{m-1}-1)}{2}$\\
$p^{m-1}(p-1)-p^{\frac{m-1}{2}}(-1)^{\frac{(p-1)(m+1)}{4}}+1$ & $\frac{p^{m-1}(p-1)^2(2p^{m-1}-1)}{2}$\\
$p^{m-1}(p-1)-p^{\frac{m-1}{2}}(-1)^{\frac{(p-1)(m+1)}{4}}+2$ & $\frac{p^{2m-2}(p-1)^3}{2}$\\
$p^{m-1}(p-1)+p^{\frac{m-1}{2}}(-1)^{\frac{(p-1)(m+1)}{4}}$ & $\frac{p^{m-1}(p-1)(p^{m-1}-1)}{2}$\\
$p^{m-1}(p-1)+p^{\frac{m-1}{2}}(-1)^{\frac{(p-1)(m+1)}{4}}+1$ & $\frac{p^{m-1}(p-1)^2(2p^{m-1}-1)}{2}$\\
$p^{m-1}(p-1)+p^{\frac{m-1}{2}}(-1)^{\frac{(p-1)(m+1)}{4}}+2$ & $\frac{p^{2m-2}(p-1)^3}{2}$\\
\hline
\end{tabular}
\end{center}
\end{table}

\begin{table}[ht]
\begin{center}
\caption{The weight distribution of $\C_{\cH(f_1)}^{(p)}$ with $m$ even}\label{tab-4}
\begin{tabular}{cc} \hline
Weight  &  Multiplicity   \\ \hline
0 & 1\\
 $p^m$ & $p-1$\\
 $p^{m-1}(p-1)$ & $p(p^{m-1}-1)$\\
 $p^{m-1}(p-1)+1$ & $p^{m}(p-1)$\\
 $p^{m-1}(p-1)+(p-1)p^{\frac{m-2}{2}}(-1)^{\frac{m(p-1)}{4}}$ & $\frac{p^{m-1}(p^{m-1}-1-(p-1)p^{\frac{m-2}{2}}(\sqrt{-1})^{\frac{(p-1)m}{2}})}{2}$\\
  $p^{m-1}(p-1)+(p-1)p^{\frac{m-2}{2}}(-1)^{\frac{m(p-1)}{4}}+1$ & $\frac{p^{m-1}(p-1)(2p^{m-1}-1-(p-2)p^{\frac{m-2}{2}}(\sqrt{-1})^{\frac{(p-1)m}{2}})}{2}$\\
    $p^{m-1}(p-1)+(p-1)p^{\frac{m-2}{2}}(-1)^{\frac{m(p-1)}{4}}+2$ & $\frac{p^{m-1}(p-1)^2(p^{m-1}+p^{\frac{m-2}{2}}(\sqrt{-1})^{\frac{(p-1)m}{2}})}{2}$\\
    $p^{m-1}(p-1)-(p-1)p^{\frac{m-2}{2}}(-1)^{\frac{m(p-1)}{4}}$ & $\frac{p^{m-1}(p^{m-1}-1+(p-1)p^{\frac{m-2}{2}}(\sqrt{-1})^{\frac{(p-1)m}{2}})}{2}$\\
  $p^{m-1}(p-1)-(p-1)p^{\frac{m-2}{2}}(-1)^{\frac{m(p-1)}{4}}+1$ & $\frac{p^{m-1}(p-1)(2p^{m-1}-1+(p-2)p^{\frac{m-2}{2}}(\sqrt{-1})^{\frac{(p-1)m}{2}})}{2}$\\
    $p^{m-1}(p-1)-(p-1)p^{\frac{m-2}{2}}(-1)^{\frac{m(p-1)}{4}}+2$ & $\frac{p^{m-1}(p-1)^2(p^{m-1}-p^{\frac{m-2}{2}}(\sqrt{-1})^{\frac{(p-1)m}{2}})}{2}$\\
 $p^{m-1}(p-1)-p^{\frac{m-2}{2}}(-1)^{\frac{m(p-1)}{4}}$ & $\frac{p^{m-1}(p-1)(p^{m-1}-1-(p-1)p^{\frac{m-2}{2}}(\sqrt{-1})^{\frac{(p-1)m}{2}})}{2}$\\
 $p^{m-1}(p-1)-p^{\frac{m-2}{2}}(-1)^{\frac{m(p-1)}{4}}+1$ & $\frac{p^{m-1}(p-1)^2(2p^{m-1}-1-(p-2)p^{\frac{m-2}{2}}(\sqrt{-1})^{\frac{(p-1)m}{2}})}{2}$\\
 $p^{m-1}(p-1)-p^{\frac{m-2}{2}}(-1)^{\frac{m(p-1)}{4}}+2$ & $\frac{p^{m-1}(p-1)^3(p^{m-1}+p^{\frac{m-2}{2}}(\sqrt{-1})^{\frac{(p-1)m}{2}})}{2}$\\
 $p^{m-1}(p-1)+p^{\frac{m-2}{2}}(-1)^{\frac{m(p-1)}{4}}$ & $\frac{p^{m-1}(p-1)(p^{m-1}-1+(p-1)p^{\frac{m-2}{2}}(\sqrt{-1})^{\frac{(p-1)m}{2}})}{2}$\\
 $p^{m-1}(p-1)+p^{\frac{m-2}{2}}(-1)^{\frac{m(p-1)}{4}}+1$ & $\frac{p^{m-1}(p-1)^2(2p^{m-1}-1+(p-2)p^{\frac{m-2}{2}}(\sqrt{-1})^{\frac{(p-1)m}{2}})}{2}$ \\
 $p^{m-1}(p-1)+p^{\frac{m-2}{2}}(-1)^{\frac{m(p-1)}{4}}+2$ & $\frac{p^{m-1}(p-1)^3(p^{m-1}-p^{\frac{m-2}{2}}(\sqrt{-1})^{\frac{(p-1)m}{2}})}{2}$\\
\hline
\end{tabular}
\end{center}
\end{table}

\begin{proof}
Let $\chi$ and $\chi'$ be the canonical additive characters of $\gf(q)$ and $\gf(p)$, respectively. Let $\eta,\eta'$ be the quadratic multiplicative characters of $\gf(q)^*$ and $\gf(p)^*$, respectively.
Denote $N_0(a,b,c)=\sharp\{x\in \gf(q):\tr_{q/p}(ax^2+bx)+c=0\}$. By the orthogonality relation of additive characters,
\begin{eqnarray}\label{eqn-4}
\nonumber N_0(a,b,c)&=&\frac{1}{p}\sum_{z\in \gf(p)}\sum_{x\in \gf(q)}\zeta_{p}^{z(\tr_{q/p}(ax^2+bx)+c)}\\
\nonumber &=&\frac{q}{p}+\frac{1}{p}\sum_{z\in \gf(p)^*}\zeta_{p}^{zc}\sum_{x\in \gf(q)}\chi(zax^2+zbx)\\
&=&p^{m-1}+\frac{1}{p}\Delta(a,b,c),
\end{eqnarray}
where $\Delta(a,b,c):=\sum_{z\in \gf(p)^*}\zeta_{p}^{zc}\sum_{x\in \gf(q)}\chi(zax^2+zbx)$. We discuss the value of $\Delta(a,b,c)$ in the following cases.
\begin{enumerate}
\item[(1)] Let $(a,b,c)=(0,0,0)$. Then $\Delta(a,b,c)=q(p-1)$.
\item[(2)] Let $(a,b)=(0,0)$ and $c\neq 0$. Then $\Delta(a,b,c)=q\sum_{z\in \gf(p)^*}\zeta_{p}^{zc}=-q$.
\item[(3)] Let $a=0$ and $b\neq 0$. Then
\begin{eqnarray*}
\Delta(a,b,c)&=&\sum_{z\in \gf(p)^*}\zeta_{p}^{zc}\sum_{x\in \gf(q)}\chi(zbx)\\
&=&0.
\end{eqnarray*}
\item[(4)] Let $a\neq 0$. By Lemma \ref{lem-charactersum-evenq} we have
\begin{eqnarray*}
\Delta(a,b,c)&=&\sum_{z\in \gf(p)^*}\zeta_{p}^{zc}\chi(-z^2b^2(4za)^{-1})\eta(za)G(\eta,\chi)\\
&=&G(\eta,\chi)\eta(a)\sum_{z\in \gf(p)^*}\zeta_{p}^{zc}\eta(z)\chi(-\frac{b^2}{4a}z)\\
&=&G(\eta,\chi)\eta(a)\sum_{z\in \gf(p)^*}\zeta_{p}^{\left(c-\tr_{q/p}(\frac{b^2}{4a})\right)z}\eta(z)\\
&=&G(\eta,\chi)\eta(a)\sum_{z\in \gf(p)^*}\chi'\left(\left(c-\tr_{q/p}(\frac{b^2}{4a})\right)z\right)\eta(z).
\end{eqnarray*}
\begin{enumerate}
\item[$\bullet$]If $m$ is odd, we have $\eta(z)=\eta'(z)$ for $z\in \gf(p)^*$. Then by Lemma \ref{quadGuasssum} we have
\begin{eqnarray*}
\lefteqn{ \Delta(a,b,c) } \\ 
&=&\left\{\begin{array}{ll}
G(\eta,\chi)\eta(a)\sum_{z\in \gf(p)^*}\eta'(z) & \mbox{ if }c=\tr_{q/p}(\frac{b^2}{4a})\\
\myatop{\mbox{$G(\eta,\chi)\eta(a)\eta'\left(c-\tr_{q/p}(\frac{b^2}{4a})\right)\times$}}
{\mbox{$\sum\limits_{z\in \gf(p)^*}\chi'\left(\left(c-\tr_{q/p}(\frac{b^2}{4a})\right)z\right)\eta'\left(\left(c-\tr_{q/p}(\frac{b^2}{4a})\right)z\right)$}} & \mbox{ if }c\neq\tr_{q/p}(\frac{b^2}{4a})\\
\end{array} \right.\\
&=&\left\{\begin{array}{ll}
0 & \mbox{ if }c=\tr_{q/p}(\frac{b^2}{4a}) \\
G(\eta,\chi)G(\eta',\chi')\eta(a)\eta'\left(c-\tr_{q/p}(\frac{b^2}{4a})\right) &
\mbox{ if }c\neq\tr_{q/p}(\frac{b^2}{4a}) \\
\end{array} \right.\\
&=&\left\{\begin{array}{ll}
0 & \mbox{ if }c=\tr_{q/p}(\frac{b^2}{4a}),\\
p^{\frac{m+1}{2}}(-1)^{\frac{(p-1)(m+1)}{4}} &
\mbox{ if }c\neq\tr_{q/p}(\frac{b^2}{4a}),\ \eta(a)\eta'\left(c-\tr_{q/p}(\frac{b^2}{4a})\right)=1,\\
p^{\frac{m+1}{2}}(-1)^{\frac{(p-1)(m+1)+4}{4}} &
\mbox{ if }c\neq\tr_{q/p}(\frac{b^2}{4a}),\ \eta(a)\eta'\left(c-\tr_{q/p}(\frac{b^2}{4a})\right)=-1.
\end{array} \right.\\
\end{eqnarray*}
\item[$\bullet$]If $m$ is even, then $\eta(z)=1$ for any $z\in \gf(p)^*$. Then by Lemma \ref{quadGuasssum} we have
\begin{eqnarray*}
\Delta(a,b,c)&=&G(\eta,\chi)\eta(a)\sum_{z\in \gf(p)^*}\chi'\left(\left(c-\tr_{q/p}(\frac{b^2}{4a})\right)z\right)\\
&=&\left\{\begin{array}{ll}
(p-1)p^{\frac{m}{2}}(-1)^{\frac{m(p-1)+4}{4}} & \mbox{ if }c=\tr_{q/p}(\frac{b^2}{4a}),\ \eta(a)=1,\\
(p-1)p^{\frac{m}{2}}(-1)^{\frac{m(p-1)}{4}} & \mbox{ if }c=\tr_{q/p}(\frac{b^2}{4a}),\ \eta(a)=-1,\\
p^{\frac{m}{2}}(-1)^{\frac{m(p-1)}{4}} &
\mbox{ if }c\neq\tr_{q/p}(\frac{b^2}{4a}),\ \eta(a)=1,\\
p^{\frac{m}{2}}(-1)^{\frac{m(p-1)+4}{4}} &
\mbox{ if }c\neq\tr_{q/p}(\frac{b^2}{4a}),\ \eta(a)=-1.
\end{array} \right.\\
\end{eqnarray*}
\end{enumerate}
\end{enumerate}
Equation (\ref{eqn-4}) and the preceding discussions yield
\begin{eqnarray}\label{eqn-5}
\lefteqn{ N_0(a,b,c)=} \nonumber \\ 
& 
\left\{\begin{array}{ll}
p^m & \mbox{ for }(a,b,c)=(0,0,0),\\
0 & \mbox{ for }(a,b)=(0,0),\ c\neq 0,\\
p^{m-1} & \mbox{ for }a=0,\ b\neq 0,\mbox{ or }a\neq 0,\ c=\tr_{q/p}(\frac{b^2}{4a}),\\
p^{m-1}+p^{\frac{m-1}{2}}(-1)^{\frac{(p-1)(m+1)}{4}} &
\mbox{ if }a\neq 0,\ c\neq\tr_{q/p}(\frac{b^2}{4a}),\ \eta(a)\eta'\left(c-\tr_{q/p}(\frac{b^2}{4a})\right)=1,\\
p^{m-1}+p^{\frac{m-1}{2}}(-1)^{\frac{(p-1)(m+1)+4}{4}} &
\mbox{ if }a\neq 0,\ c\neq\tr_{q/p}(\frac{b^2}{4a}),\ \eta(a)\eta'\left(c-\tr_{q/p}(\frac{b^2}{4a})\right)=-1,
\end{array} \right.
\end{eqnarray}
if $m$ is odd, and
\begin{eqnarray}\label{eqn-6}
N_0(a,b,c)=
\left\{\begin{array}{ll}
p^m & \mbox{ for }(a,b,c)=(0,0,0),\\
0 & \mbox{ for }(a,b)=(0,0),\ c\neq 0,\\
p^{m-1} & \mbox{ for }a=0,\ b\neq 0,\\
p^{m-1}+(p-1)p^{\frac{m-2}{2}}(-1)^{\frac{m(p-1)+4}{4}} & \mbox{ if }a\neq 0,\ c=\tr_{q/p}(\frac{b^2}{4a}),\ \eta(a)=1,\\
p^{m-1}+(p-1)p^{\frac{m-2}{2}}(-1)^{\frac{m(p-1)}{4}} & \mbox{ if }a\neq 0,\ c=\tr_{q/p}(\frac{b^2}{4a}),\ \eta(a)=-1,\\
p^{m-1}+p^{\frac{m-2}{2}}(-1)^{\frac{m(p-1)}{4}} &
\mbox{ if }a\neq 0,\ c\neq\tr_{q/p}(\frac{b^2}{4a}),\ \eta(a)=1,\\
p^{m-1}+p^{\frac{m-2}{2}}(-1)^{\frac{m(p-1)+4}{4}} &
\mbox{ if }a\neq 0,\ c\neq\tr_{q/p}(\frac{b^2}{4a}),\ \eta(a)=-1,
\end{array} \right.
\end{eqnarray}
if $m$ is even.

For any codeword $\bc(a,b,c)=\left((\tr_{q/p}(ax^2+bx)+c)_{x\in \gf(q)},\tr_{q/p}(a),\tr_{q/p}(b)\right)$ with $m$ odd, by Equation (\ref{eqn-5}) we easily derive that
\begin{eqnarray*}
\lefteqn{ \wt(\bc(a,b,c)) } \\ 
&=&\left\{\begin{array}{ll}
0 & \mbox{ for }(a,b,c)=(0,0,0) \\
p^m & \mbox{ for }(a,b)=(0,0),\ c\neq 0 \\
p^{m-1}(p-1) & \mbox{ for }a=0,\ b\neq 0,\ \tr_{q/p}(b)=0 \\
p^{m-1}(p-1)+1 & \mbox{ for }a=0,\ b\neq 0,\ \tr_{q/p}(b)\neq0 \\
p^{m-1}(p-1) & \myatop{\mbox{ for $a\neq 0,\ c=\tr_{q/p}(\frac{b^2}{4a}),$}}{\mbox{$\tr_{q/p}(a)=\tr_{q/p}(b)=0$}} \\
p^{m-1}(p-1)+1 & \myatop{\mbox{ for $a\neq 0,\ c=\tr_{q/p}(\frac{b^2}{4a}),$}}{\mbox{$\tr_{q/p}(a)\neq 0,
\ \tr_{q/p}(b)=0$}} \\
p^{m-1}(p-1)+1 & \myatop{\mbox{ for $a\neq 0,\ c=\tr_{q/p}(\frac{b^2}{4a}),$}}{\mbox{$\tr_{q/p}(a)=0,
\ \tr_{q/p}(b)\neq0$}} \\
p^{m-1}(p-1)+2 & \myatop{\mbox{ for $a\neq 0,\ c=\tr_{q/p}(\frac{b^2}{4a}),$}}{\mbox{$\tr_{q/p}(a)\neq 0,
\ \tr_{q/p}(b)\neq0$}} \\
p^{m-1}(p-1)-p^{\frac{m-1}{2}}(-1)^{\frac{(p-1)(m+1)}{4}} &
\myatop{\mbox{ for $a\neq 0,\ c\neq\tr_{q/p}(\frac{b^2}{4a}),$}}{\myatop{\mbox{$\eta(a)\eta'\left(c-\tr_{q/p}(\frac{b^2}{4a})\right)=1,$}}{\mbox{$\tr_{q/p}(a)=\tr_{q/p}(b)=0$}}} \\
p^{m-1}(p-1)-p^{\frac{m-1}{2}}(-1)^{\frac{(p-1)(m+1)}{4}}+1 &
\myatop{\mbox{ for $a\neq 0,\ c\neq\tr_{q/p}(\frac{b^2}{4a}),$}}{\myatop{\mbox{$\eta(a)\eta'\left(c-\tr_{q/p}(\frac{b^2}{4a})\right)=1,$}}{\mbox{$\tr_{q/p}(a)=0,\ \tr_{q/p}(b)\neq0$}}} \\
p^{m-1}(p-1)-p^{\frac{m-1}{2}}(-1)^{\frac{(p-1)(m+1)}{4}}+1 &
\myatop{\mbox{ for $a\neq 0,\ c\neq\tr_{q/p}(\frac{b^2}{4a}),$}}{\myatop{\mbox{$\eta(a)\eta'\left(c-\tr_{q/p}(\frac{b^2}{4a})\right)=1,$}}{\mbox{$\tr_{q/p}(a)\neq 0,\ \tr_{q/p}(b)=0$}}} \\
p^{m-1}(p-1)-p^{\frac{m-1}{2}}(-1)^{\frac{(p-1)(m+1)}{4}}+2 &
\myatop{\mbox{ for $a\neq 0,\ c\neq\tr_{q/p}(\frac{b^2}{4a}),$}}{\myatop{\mbox{$\eta(a)\eta'\left(c-\tr_{q/p}(\frac{b^2}{4a})\right)=1,$}}{\mbox{$\tr_{q/p}(a)\neq 0,\ \tr_{q/p}(b)\neq 0$}}} \\
p^{m-1}(p-1)-p^{\frac{m-1}{2}}(-1)^{\frac{(p-1)(m+1)+4}{4}} &
\myatop{\mbox{ for $a\neq 0,\ c\neq\tr_{q/p}(\frac{b^2}{4a}),$}}{\myatop{\mbox{$\eta(a)\eta'\left(c-\tr_{q/p}(\frac{b^2}{4a})\right)=-1,$}}{\mbox{$\tr_{q/p}(a)=\tr_{q/p}(b)=0$}}} \\
p^{m-1}(p-1)-p^{\frac{m-1}{2}}(-1)^{\frac{(p-1)(m+1)+4}{4}}+1 &
\myatop{\mbox{ for $a\neq 0,\ c\neq\tr_{q/p}(\frac{b^2}{4a}),$}}{\myatop{\mbox{$\eta(a)\eta'\left(c-\tr_{q/p}(\frac{b^2}{4a})\right)=-1,$}}{\mbox{$\tr_{q/p}(a)=0,\ \tr_{q/p}(b)\neq0$}}}\\
p^{m-1}(p-1)-p^{\frac{m-1}{2}}(-1)^{\frac{(p-1)(m+1)+4}{4}}+1 &
\myatop{\mbox{ for $a\neq 0,\ c\neq\tr_{q/p}(\frac{b^2}{4a}),$}}{\myatop{\mbox{$\eta(a)\eta'\left(c-\tr_{q/p}(\frac{b^2}{4a})\right)=-1,$}}{\mbox{$\tr_{q/p}(a)\neq 0,\ \tr_{q/p}(b)=0$}}}\\
p^{m-1}(p-1)-p^{\frac{m-1}{2}}(-1)^{\frac{(p-1)(m+1)+4}{4}}+2 &
\myatop{\mbox{ for $a\neq 0,\ c\neq\tr_{q/p}(\frac{b^2}{4a}),$}}{\myatop{\mbox{$\eta(a)\eta'\left(c-\tr_{q/p}(\frac{b^2}{4a})\right)=-1,$}}{\mbox{$\tr_{q/p}(a)\neq 0,\ \tr_{q/p}(b)\neq 0$}}}\\
\end{array} \right.\\
&=&\left\{\begin{array}{ll}
0 & \mbox{ with 1 time,}\\
p^m & \mbox{ with $p-1$ times,}\\
p^{m-1}(p-1) & \mbox{ with $(p^{m-1}-1)(p+p^{m-1})$ times,}\\
p^{m-1}(p-1)+1 & \mbox{ with $(p^{m}-p^{m-1})(p+2p^{m-1}-1)$ times,}\\
p^{m-1}(p-1)+2 & \mbox{ with $(p^{m}-p^{m-1})^2$ times,}\\
p^{m-1}(p-1)-p^{\frac{m-1}{2}}(-1)^{\frac{(p-1)(m+1)}{4}} & \mbox{ with $\frac{p^{m-1}(p-1)(p^{m-1}-1)}{2}$ times,}\\
p^{m-1}(p-1)-p^{\frac{m-1}{2}}(-1)^{\frac{(p-1)(m+1)}{4}}+1 & \mbox{ with $\frac{p^{m-1}(p-1)^2(2p^{m-1}-1)}{2}$ times,}\\
p^{m-1}(p-1)-p^{\frac{m-1}{2}}(-1)^{\frac{(p-1)(m+1)}{4}}+2 & \mbox{ with $\frac{p^{2m-2}(p-1)^3}{2}$ times,}\\
p^{m-1}(p-1)-p^{\frac{m-1}{2}}(-1)^{\frac{(p-1)(m+1)+4}{4}} & \mbox{ with $\frac{p^{m-1}(p-1)(p^{m-1}-1)}{2}$ times,}\\
p^{m-1}(p-1)-p^{\frac{m-1}{2}}(-1)^{\frac{(p-1)(m+1)+4}{4}}+1 & \mbox{ with $\frac{p^{m-1}(p-1)^2(2p^{m-1}-1)}{2}$ times,}\\
p^{m-1}(p-1)-p^{\frac{m-1}{2}}(-1)^{\frac{(p-1)(m+1)+4}{4}}+2 & \mbox{ with $\frac{p^{2m-2}(p-1)^3}{2}$ times,}\\
\end{array} \right.\\
\end{eqnarray*}
where the frequency is easy to obtain based on Lemma \ref{lem-equalities}. As an example, we compute the frequency of the nonzero weight
$$w=p^{m-1}(p-1)-p^{\frac{m-1}{2}}(-1)^{\frac{(p-1)(m+1)}{4}},$$ i.e.
$$A_w=\sharp \{(a,b,c):a\neq 0,\ c\neq\tr_{q/p}(\frac{b^2}{4a}),\ \eta(a)\eta'\left(c-\tr_{q/p}(\frac{b^2}{4a})\right)=1,\ \tr_{q/p}(a)=\tr_{q/p}(b)=0\}.$$
 Clearly, the number of $b$ such that $\tr_{q/p}(b)=0$ is $n_b=p^{m-1}$.  For $\eta(a)=\eta'\left(c-\tr_{q/p}(\frac{b^2}{4a})\right)=1$, the number of $a$ such that $\eta(a)=1,\tr_{q/p}(a)=0$ equals $n_a=\frac{p^{m-1}-1}{2}$ by Lemma \ref{lem-equalities}, and if we fix $a,b$, then the number of $c$ such that $\eta'\left(c-\tr_{q/p}(\frac{b^2}{4a})\right)=1$ is equal to $n_c=\frac{p-1}{2}$. For $\eta(a)=\eta'\left(c-\tr_{q/p}(\frac{b^2}{4a})\right)=-1$, the number of $a$ such that $\eta(a)=-1,\tr_{q/p}(a)=0$ equals $n_a=\frac{p^{m-1}-1}{2}$ by Lemma \ref{lem-equalities}, and if we fix $a,b$, then the number of $c$ such that $\eta'\left(c-\tr_{q/p}(\frac{b^2}{4a})\right)=-1$ is equal to $n_c=\frac{p-1}{2}$. Hence,
$$A_w=2n_an_bn_c=\frac{p^{m-1}(p-1)(p^{m-1}-1)}{2}.$$
The frequencies of other nonzero weights can be similarly derived. The dimension is $2m+1$ as $A_0=1$.

For even $m$, we can similarly derive the weight distribution by Equation (\ref{eqn-6}) and Lemma \ref{lem-equalities}. We omit the details here.
The desired conclusions then follow.
\end{proof}

We remark that $\C_{\cH(f_1)}^{(p)}$ in Theorem \ref{th-main3} produces optimal or almost optimal codes in some cases according to the tables of best codes known
maintained at http://www.codetables.de. These optimal or almost optimal codes are listed in Table \ref{tab-52}. Here we call an $[n,k,d]$ code almost optimal if  the corresponding optimal code has parameters $[n,k,d+1]$. 

Note that the function $f_1(x)=x^2$ over $\gf(p^m)$ is linear when $p=2$ and nonlinear when 
$p$ is odd. Hence, $\C_{\cH(f_1)}^{(2)}$ and $\C_{\cH(f_1)}^{(p)}$ for odd $p$ are different 
in several aspects. The dimensions of the codes in the two cases tell us such difference.

\begin{table}[ht]
\begin{center}
\caption{Optimal or almost optimal codes produced by  $\C_{\cH(f_1)}^{(p)}$}\label{tab-52}
\begin{tabular}{cccc} \hline
$p$  &  $m$ & Parameters & Optimality   \\ \hline
3 & 1 & $[5,3,2]$ & Optimal \\
5 & 1 & $[7,3,4]$ & Optimal \\
7 & 1 & $[9,3,6]$ & Optimal \\
3 & 3 & $[29,7,15]$ & Optimal \\
3 & 4 & $[83,9,48]$ & Almost optimal \\
5 & 3 & $[127,7,95]$ & Optimal \\
\hline
\end{tabular}
\end{center}
\end{table}

\section{The subfield codes of the conic codes}\label{sec-coniccode}

Let $q$ be odd. We consider the conic in $\PG(2,\gf(q))$: 
$$\cO=\{(x^2,x,1):x\in \gf(q)\}\cup \{(1,0,0)\}.$$ 
The conic code $\C_{\cO}$ over $\gf(q)$ has parameters $[q+1,3,q-1]$  and generator matrix
$$G_{\cO}=\begin{bmatrix} x_1^2 & x_2^2 & \cdots & x_q^2 & 1\\
 x_1 & x_2 & \cdots & x_q & 0\\
 1 & 1 & \cdots & 1 & 0\\ \end{bmatrix},$$ where $\gf(q)=\{x_1,x_2,\cdots,x_{q}\}$. In this section, we study the subfield code $\C_{\cO}^{(p)}$ of $\C_{\cO}$ with $p$ odd.

By Lemma \ref{th-tracerepresentation}, the trace representation of $\C_{\cO}^{(p)}$ is given by
\begin{eqnarray*}
\C_{\cO}^{(p)}=\left\{\left((\tr_{q/p}(ax^2+bx)+c)_{x\in \gf(q)},\tr_{q/p}(a)\right):a,b\in \gf(q),c\in \gf(p)\right\}.
\end{eqnarray*}
Note that  $\C_{\cO}^{(p)}$ can be viewed as a punctured code from $\C_{\cH(f_1)}^{(p)}$ in Section \ref{sect-p-ary-case}.

\begin{theorem}\label{th-main4}
Let $q=p^m$ with $p$ odd and $m>1$ a positive integer. The following statements hold.
\begin{enumerate}
\item[(1)]If $m$ is odd, then $\C_{\cO}^{(p)}$ is a $[p^m+1,2m+1,p^{m-1}(p-1)-p^{\frac{m-1}{2}}]$ code with the weight distribution in Table \ref{tab-5}.
\item[(2)] If $m$ is even, then $\C_{\cO}^{(p)}$ is a $[p^m+1,2m+1,p^{m-1}(p-1)-(p-1)p^{\frac{m-2}{2}}]$  code with the weight distribution in Table \ref{tab-6}.
\item[(3)] If $p>3$, then $\C_{\cO}^{(p)\bot}$ is a $[p^m+1,p^m-2m,4]$  code which is distance-optimal with respect to the sphere-packing bound.
\end{enumerate}
\end{theorem}

\begin{proof}
Since $\C_{\cO}^{(p)}$ is a punctured code of $\C_{\cH(f_1)}^{(p)}$, the weight distribution 
of $\C_{\cO}^{(p)}$ follows from the proof of Theorem \ref{th-main3}. We omit the details of 
the proof here. 
\end{proof}

\begin{table}[ht]
\begin{center}
\caption{The weight distribution of $\C_{\cO}^{(p)}$ with $m$ odd}\label{tab-5}
\begin{tabular}{cc} \hline
Weight  &  Multiplicity   \\ \hline
0 & 1 \\
$p^m$ & $p-1$\\
$p^{m-1}(p-1)$ & $p(p^m-1)+p^m(p^{m-1}-1)$\\
$p^{m-1}(p-1)+1$ & $p^m(p^m-p^{m-1})$\\
$p^{m-1}(p-1)-p^{\frac{m-1}{2}}(-1)^{\frac{(p-1)(m+1)}{4}}$ & $\frac{p^m(p^{m-1}-1)(p-1)}{2}$\\
$p^{m-1}(p-1)-p^{\frac{m-1}{2}}(-1)^{\frac{(p-1)(m+1)}{4}}+1$ & $\frac{p^{2m-1}(p-1)^2}{2}$\\
$p^{m-1}(p-1)+p^{\frac{m-1}{2}}(-1)^{\frac{(p-1)(m+1)}{4}}$ & $\frac{p^m(p^{m-1}-1)(p-1)}{2}$\\
$p^{m-1}(p-1)+p^{\frac{m-1}{2}}(-1)^{\frac{(p-1)(m+1)}{4}}+1$ & $\frac{p^{2m-1}(p-1)^2}{2}$\\
\hline
\end{tabular}
\end{center}
\end{table}

\begin{table}[ht]
\begin{center}
\caption{The weight distribution of $\C_{\cO}^{(p)}$ with $m$ even}\label{tab-6}
\begin{tabular}{cc} \hline
Weight  &  Multiplicity   \\ \hline
0 & 1\\
 $p^m$ & $p-1$\\
 $p^{m-1}(p-1)$ & $p(p^{m}-1)$\\
 $p^{m-1}(p-1)+(p-1)p^{\frac{m-2}{2}}(-1)^{\frac{m(p-1)}{4}}$ & $\frac{p^{m}(p^{m-1}-1-(p-1)p^{\frac{m-2}{2}}(\sqrt{-1})^{\frac{(p-1)m}{2}})}{2}$\\
  $p^{m-1}(p-1)+(p-1)p^{\frac{m-2}{2}}(-1)^{\frac{m(p-1)}{4}}+1$ & $\frac{p^{m}(p-1)(p^{m-1}+p^{\frac{m-2}{2}}(\sqrt{-1})^{\frac{(p-1)m}{2}})}{2}$\\
    $p^{m-1}(p-1)-(p-1)p^{\frac{m-2}{2}}(-1)^{\frac{m(p-1)}{4}}$ & $\frac{p^{m}(p^{m-1}-1+(p-1)p^{\frac{m-2}{2}}(\sqrt{-1})^{\frac{(p-1)m}{2}})}{2}$\\
  $p^{m-1}(p-1)-(p-1)p^{\frac{m-2}{2}}(-1)^{\frac{m(p-1)}{4}}+1$ & $\frac{p^{m}(p-1)(p^{m-1}-p^{\frac{m-2}{2}}(\sqrt{-1})^{\frac{(p-1)m}{2}})}{2}$\\
 $p^{m-1}(p-1)-p^{\frac{m-2}{2}}(-1)^{\frac{m(p-1)}{4}}$ & $\frac{p^{m}(p-1)(p^{m-1}-1-(p-1)p^{\frac{m-2}{2}}(\sqrt{-1})^{\frac{(p-1)m}{2}})}{2}$\\
 $p^{m-1}(p-1)-p^{\frac{m-2}{2}}(-1)^{\frac{m(p-1)}{4}}+1$ & $\frac{p^{m}(p-1)^2(p^{m-1}+p^{\frac{m-2}{2}}(\sqrt{-1})^{\frac{(p-1)m}{2}})}{2}$\\
 $p^{m-1}(p-1)+p^{\frac{m-2}{2}}(-1)^{\frac{m(p-1)}{4}}$ & $\frac{p^{m}(p-1)(p^{m-1}-1+(p-1)p^{\frac{m-2}{2}}(\sqrt{-1})^{\frac{(p-1)m}{2}})}{2}$\\
 $p^{m-1}(p-1)+p^{\frac{m-2}{2}}(-1)^{\frac{m(p-1)}{4}}+1$ & $\frac{p^{m}(p-1)^2(p^{m-1}-p^{\frac{m-2}{2}}(\sqrt{-1})^{\frac{(p-1)m}{2}})}{2}$ \\
\hline
\end{tabular}
\end{center}
\end{table}

We remark that $\C_{\cO}^{(p)}$ in Theorem \ref{th-main4} produces optimal or almost optimal codes in some cases according to the tables of best codes known
maintained at http://www.codetables.de. These optimal or almost optimal codes are listed in Table \ref{tab-7}. Here we call an $[n,k,d]$ code almost optimal if  the corresponding optimal code has parameters $[n,k,d+1]$.
\begin{table}[ht]
\begin{center}
\caption{Optimal or almost optimal codes produced by  $\C_{\cO}^{(p)}$}\label{tab-7}
\begin{tabular}{cccc} \hline
$p$  &  $m$ & Parameters & Optimality   \\ \hline
3 & 2 & $[10,5,4]$ & Almost optimal \\
3 & 3 & $[28,7,15]$ & Optimal \\
3 & 4 & $[82,9,48]$ & Optimal \\
5 & 2 & $[26,5,16]$ & Almost optimal \\
5 & 3 & $[126,7,95]$ & Optimal \\
\hline
\end{tabular}
\end{center}
\end{table}

\section{Conclusions and remarks}\label{sect-conlusion}

In this paper, we mainly investigated the binary subfield codes of two classes of hyperoval codes and the $p$-ary subfield codes of the conic codes for odd $p$. The weight distribution of the $p$-ary code $\C_{\cH(f_1)}^{(p)}$ was also determined. The codes presented in this paper have optimal or almost optimal parameters in some cases. These codes are summarized as follows:
\begin{enumerate}
\item[$\bullet$] Examples \ref{exa-1} and \ref{exa-2} show that the binary subfield codes of the hyperoval codes are optimal in some cases.
\item[$\bullet$] Table \ref{tab-5} demonstrates that the $p$-ary code $\C_{\cH(f_1)}^{(p)}$ is optimal or almost optimal in some cases. 
\item[$\bullet$] Table \ref{tab-7} demonstrates that the $p$-ary subfield code $\C_{\cO}^{(p)}$ of the conic code with $p$ odd is optimal or almost optimal in some cases.
\item[$\bullet$] In Theorem \ref{th-main3}, a class of $p$-ary almost MDS $[p+2,3,p-1]$ code is presented.
\item[$\bullet$] Note that the subfield code $\C_{\cH(f_1)}^{(2)}$ has bad parameters. However, 
its dual code $\C_{\cH(f_1)}^{(2)\bot}$ is distance-optimal according to the sphere-packing bound and has parameters $[2^m+2,2^m-m,4]$. This shows that it is worthwhile to study the 
subfield code $\C_{\cH(f_1)}^{(2)}$. 
\item[$\bullet$] Theorem \ref{th-main4} shows that $\C_{\cO}^{(p)\bot}$ is a $[p^m+1,p^m-2m,4]$ code which is distance-optimal with respect to the sphere-packing bound for $p>3$.
\end{enumerate} 
Our motivation of studying the subfield codes of hyperoval and conic codes is to 
construct linear codes with new parameters. To our knowledge, these codes have 
new parameters as families of linear codes. 

We mention that hyperovals are related to bent functions \cite{AM17,CM16}. Bent 
functions could be employed to construct linear codes in many ways \cite{Ding15,Mesn15,Mesn,TCZ}. 
However, no codes of length $q^m+1$ are constructed in these references. Hence, the codes 
presented in this paper would be different.  
 
In addition to the translation and Segre hyperovals, there are several other families of 
hyperovals in the literature \cite{DY}. It seems hard to determine the minimum distances 
of the subfield codes of these hyperoval codes, let alone their weight distributions. 
The reader is cordially invited to settle these problems.

\end{document}